\documentclass[10pt]{article}
\usepackage[utf8]{inputenc}
\usepackage[natbibapa]{apacite}
\usepackage{xcolor}
\usepackage{graphicx}
\usepackage{epsfig}
\usepackage{multirow}
\usepackage{colortbl}
\usepackage{changepage}
\usepackage{multirow}
\usepackage{amsfonts,amssymb,graphics,epsfig,verbatim,bm,latexsym,amsmath,amsbsy,bbm,bm,amsthm}
\usepackage{todonotes}
\usepackage{pseudocode}
\usepackage{algorithm}\usepackage{float}\usepackage{lscape}
\usepackage{soul}
\usepackage{algpseudocode}\usepackage{capt-of}

\DeclareMathOperator*{\argminA}{arg\,min} 

\newtheorem{assumption}{Assumption}
\newtheorem{proposition}{Proposition}

\newtheorem{remark}{Remark}

\title{Optimal selection of the number of control units in kNN algorithm to estimate average treatment effects}

\author{Andr\'es Ram\'irez-Hassan\thanks{Department of Economics, Universidad EAFIT, Medell\'in, Colombia; email: aramir21@eafit.edu.co} \and Raquel Vargas-Correa\thanks{Department of Economics, Universidad EAFIT, Medell\'in, Colombia; email: rvargas@eafit.edu.co} \and Gustavo Garcia\thanks{Department of Economics, Universidad EAFIT, Medell\'in, Colombia; email: ggarci24@eafit.edu.co} \and Daniel Londoño\thanks{Department of Economics, Universidad EAFIT, Medell\'in, Colombia; email: dlondoko@eafit.edu.co}}
	
\date{\today}

\begin{document}
\maketitle \thispagestyle{empty}

\begin{abstract}
\setlength{\baselineskip}{10pt}
\vspace{0.2in} \noindent We propose a simple approach to optimally select the number of control units in $k$ nearest neighbors (kNN) algorithm focusing in minimizing the mean squared error for the average treatment effects. Our approach is non-parametric where confidence intervals for the treatment effects were calculated using asymptotic results with bias correction. Simulation exercises show that our approach gets relative small mean squared errors, and a balance between confidence intervals length and type I error. We analyzed the average treatment effects on treated (ATET) of participation in 401(k) plans on accumulated net financial assets confirming significant effects on amount and positive probability of net asset. Our optimal $k$ selection produces significant narrower ATET confidence intervals compared with common practice of using $k=1$. 


\vspace{0.2in} \noindent
{\normalsize JEL Classification: C14, C18, O18, R42. 
}

\end{abstract}
\newpage
\setcounter{page}{1}
\pagenumbering{arabic}

\section{Introduction}
Matching estimators are commonly used to estimate average treatment effects when treatment assignment is independent of outcome conditional on confounding variables \citep[p.~871]{Cameron05}. However, it seems that optimal selection of the number of matching units ($k$) is an open question as most of the applied literature fixes $k$ small without any optimal criterion, often $k=1$ \citep{Abadie2016}. Choice of $k$ is a trade-off between bias and variance. Small $k$ implies small bias, but high variance, and vice versa. So, the aim of this paper is to propose a simple approach to optimally select $k$. In particular, we focus on the $k$ nearest-neighbors (kNN) algorithm when used to estimate treatment effects: average treatment effect (ATE), and average treatment effect on treated (ATET).

Optimal selection of $k$ in kNN algorithm is done by cross validation when prediction or classification are the objectives \citep[p.~241]{Hastie2009}. This is easily performed as the statistical object of interest is observable in-sample. On the other hand, average treatment effects are not observable anywhere. So, we follow similar arguments to \cite{Athey2015} proposing an unbiased function of observable variables to optimally choose $k$ through cross validation when the main inferential concern is treatment effects. 

We apply our proposal to analyze ATET due to enrollment in 401(k) on amount and probability of positive accumulated net financial assets. We found significant ATET whose point estimates are approximately \$15 K and 19\% given an optimal $k^*=19$ in both cases.

After this brief introduction, Section \ref{Optimal} shows our proposal to optimally select $k$ in the kNN algorithm. Section \ref{Simulations} displays results from some simulations exercises, and Section \ref{Application} shows the results of our application regarding 401(k) enrollment. Section \ref{Conclusions} ends with some concluding remarks.

\section{Optimal selection of $k$ in kNN}\label{Optimal}

We will focus on the average treatment effect on treated as we will use this in our application. However, results associated with the average treatment effects are in the Appendix (see subsection \ref{ref:ATE}). We build the matching units using $k$ nearest neighbors (kNN) with replacement, where $k$ is selected minimizing the mean squared error between a conditional unbiased estimator of the average treatment effect on the treated (ATET) and a \textit{matching estimator} of the ATET.

We consider a framework where there is a random binary treatment ($D_i=\left\{0,1\right\}$) that is independent of the outcome variable ($Y_i$) conditional on independent observable variables ($\bm{X}_i\subset \mathbb{R}^P$), $i=1,2,\dots,N$. The average treatment effect on treated,

\begin{equation}\label{eq1}
\Delta^{T}=\mathbb{E}\left[(Y_i(1)-Y_i(0))|D_i=1, \bm{X}_i=\bm{x}_i\right],   
\end{equation}

\noindent where $\bm{X}_i$ are control variables, $\left\{Y_i(0),Y_i(1)\right\}$ are the potential outcomes under different states, untreated and treated,
\begin{align*}
    Y_i=Y_i(D_i)=
    \begin{Bmatrix}
    Y_i(0), & D_i=0\\
    Y_i(1), & D_i=1
    \end{Bmatrix}.
\end{align*}

However, we do not simultaneously observe the same individual under both states. So, we propose to use the matching approach to build the synthetic controls, the general formula is \citep[p.~875]{Cameron05},

\begin{equation}\label{eq2A}
    \hat{\Delta}^{TM}=\frac{1}{N_1}\sum_{i\in\left\{D_i=1\right\}}\left[Y_i(1)-\frac{1}{k}\sum_{j\in A_k^i(\bm{x})}Y_j(0)\right],
\end{equation}

where $A_k^i(\bm{x})=\left\{j:||\bm{X}_i-\bm{X}_j||<||\bm{X}_i-\bm{X}_l||, j\in\left\{D_j=0\right\}, l\in\left\{D_l=0\right\}\right\}$ is the set of the $k$ closest untreated units to treated unit $i$ in terms of the Euclidean norm in covariates, $j,l=1,2,\dots,N_0$, $||\cdot ||$ is the Euclidean norm, $N_0$ is the sample size of untreated individuals ($D_i=0$), $N_1$ is the sample size of treated individuals ($D_i=1$), and $k$ is the cardinality of $A_k^i(\bm{x})$, that is, the number of controls for treated unit,
\begin{equation*}
    |A_k^i(\bm{x})|=\sum_{j\in A_k^i(\bm{x})}\mathbbm{1}\left\{||\bm{X}_i-\bm{X}_j||<||\bm{X}_i-\bm{X}_l||\right\}=k.    
\end{equation*}

Therefore, the synthetic control for treated individual $i$ is built using an average of untreated nearest individuals (neighbors).

The choice of $k$ is a trade-off between bias and variance. Selecting just one neighbor minimizes bias, this is equal to zero if there is exact match ($k=1, \bm{x_i}=\bm{x_j}$), but implies high variability. On the other hand, a large amount of neighbors decreases variance but increases bias. In general, $\bm{x_i}=\bm{x_j}$ is an event of probability 0 for continuous covariates, then \cite{Abadie2006} show that matching estimators have an asymptotic bias. This asymptotic bias can be ignored if $N_1^{P_c/2}=O(N_0)$ where $P_c$ is the number of continuous covariates ($P_c/2>1$). The bias of the ATET is 

\begin{equation}\label{eq3A} 
    B^{TM}=\frac{1}{N_1}\sum_{i\in\left\{D_i=1\right\}}\left[\mu_0(\bm{x}_i)-\frac{1}{k}\sum_{j\in A_k^i(\bm x)}\mu_0(\bm{x}_j)\right],
\end{equation}

\noindent where $\mu_0(\bm{x}_i)=\mathbb{E}\left[Y|\bm{X}=\bm{x},D=0\right]$. This can be estimated non-parametrically using a series expansion estimator to obtain $\hat{B}^{TM}$ \citep{Abadie2011}.

Therefore by equations \ref{eq2A} and \ref{eq3A}, the bias-corrected matching estimator is

\begin{equation}\label{eqBC_ATET}
    \hat{\Delta}^{BC}=\hat{\Delta}^{TM}-\hat{B}^{TM}.
\end{equation}

\cite{Abadie2006} show under suitable assumptions (see Assumptions 1, 2', 3' and 4 in their paper) that 

\begin{equation*}
    (V^{E,T}+V^{\Delta(\bm{X}),T})^{-1/2}\sqrt{N_1}(\hat{\Delta}^{TM}-B^{TM}-\Delta^T)\xrightarrow[]{d}N(0,1),
\end{equation*}
\noindent where $V^{E,T}=\frac{1}{N_1}\sum_{i=1}^N \left(D_i+(1-D_i)\frac{J_k^i}{k}\right)^2 \sigma^2(\bm{X}_i,D_i)$, $J_k^i=\sum_{j=1}^N\mathbbm{1}\left\{i\in A_k^j\right\}$ is the number of times unit $i$ is used as a match given $k$ matches per unit, $\sigma^2(\bm{X}_i,D_i)=\mathbb{V}(Y|\bm{X}_i=\bm{x}_i,D_i=d_i)$, which can be consistently estimated by $\hat{\sigma}^2(\bm{X}_i,D_i)=\frac{k}{k+1}\left(Y_i-\frac{1}{k}\sum_{l\in L_k^i}\right)$, $L_k^i(\bm{x})$ is the set of the $k$ closest units to unit $i$ in terms of the Euclidean norm in covariates such that have the same value for the treatment variable, and $V^{\Delta(\bm{X}),T}=\mathbb{E}\left[(\left[\mu_1(\bm{X}_i)-\mu_0(\bm{X}_i)\right]-{\Delta}^{T})^2|\bm{D}=1\right]$ is the variance of the conditional average treatment effect on treated.

In an ideal situation, we select $k$ in the regression nearest neighbor framework (kNN) minimizing the mean squared error of the average treatment effect on treated (equation \ref{eq1}), that is,

\begin{equation}\label{eq2}
    \argminA_{k} \mathbb{E}\left({\hat{\Delta}^{BC}}-{\Delta}^T\right)^2.
\end{equation}

However, the average treatment effect on treated ($\Delta^{T}$) is not observed. Therefore, we follow similar arguments to \cite{Athey2015} proposing an unbiased function of potentially observable variables to define the ATET in program \ref{eq2}.

We have the following assumptions:

\begin{assumption}\label{assump1} Stable unit treatment value assumption (SUTVA)
\end{assumption}

\begin{assumption}\label{assump2} Ignorability (unconfoundedness)
    \begin{equation*}
        D_i\perp Y_i(0) | \bm{X}_i=\bm{x}_i.   
    \end{equation*}
\end{assumption}

\begin{assumption}\label{assump3} Overlap (matching)
    \begin{equation*}
        0\leq P(D_i=1 | \bm{X}_i=\bm{x}_i)<1.    
    \end{equation*}
\end{assumption}

\begin{remark}
Assumption \ref{assump1} implies not spillover, interaction or general equilibrium effects \citep{Rubin1978}. Assumption \ref{assump2} establishes that treatment assignment ignores untreated outcome given control variables \citep{Rubin1978,Rosenbaum1983}. Assumption \ref{assump3} says that it is necessary to have overlap in subsamples of $\bm{X}_i$, that is, for each treated individual it is necessary to have an analogous untreated (synthetic control) to identify the treatment effect on treated \citep{Rosenbaum1983}.
\end{remark}

We assume conditional independent and identically distributed random sample drawn from a population. This implicitly implies Assumption \ref{assump1} (SUTVA). Then, we define the \textit{individual treatment effect on treated} 

\begin{equation}\label{eq3}
    {\Delta_i^{T*}}=\frac{Y_i(D_i-P(\bm{X}_i))}{P\left[D=1\right](1-P(\bm{X}_i))},
\end{equation}

\noindent where $P\left[D=1\right]$ is the marginal probability of treatment, which is the same for every individual given random sampling, and $P(\bm{X}_i)=P\left[D_i=1|\bm{X}_i=\bm{x}\right]$ is the conditional probability of treatment given regressors $\bm{X}_i$, that is, the propensity score \citep{Rosenbaum1983}.

\begin{proposition}\label{prop1}
    Assuming a conditional independent and identically distributed random sample such that Assumptions \ref{assump2} and \ref{assump3} are satisfied, then:
    \begin{equation*}
        \mathbb{E}\left[{\Delta_i^{T*}}\big| P(\bm{X}_i),\bm{X}_i=\bm{x}\right]=\Delta^{T}.
    \end{equation*}
\end{proposition}

\begin{proof}
We omit conditioning on $P\left[D=1\right]$ and $\bm{X}_i=\bm{x}$ to simplify notation. 
\begin{align*}
\mathbb{E}\left[{\Delta_i^{T*}}\big| P\left[D=1\right],\bm{X}_i=\bm{x}\right]=&\mathbb{E}\left[\frac{Y_i(D_i-P(\bm{X}_i))}{P\left[D=1\right](1-P(\bm{X}_i))}\right]\\
=&\frac{1}{P\left[D=1\right] (1-P(\bm{X}_i))}\left\{(1-P(\bm{X}_i))\mathbb{E}\left[D_i(Y_i(1)-Y_i(0))\right] \right.\\
& \left.+ \mathbb{E}\left[Y_i(0)(D_i-P(\bm{X}_i))\right]\right\}\\
=&\frac{1}{P\left[D=1\right] (1-P(\bm{X}_i))}\left\{(1-P(\bm{X}_i))\mathbb{E}\left[\mathbb{E}\left[D_i(Y_i(1)-Y_i(0))\big|D_i\right]\right] \right.\\
& \left. + \mathbb{E}\left[D_i\right]\mathbb{E}\left[Y_i(0)\right]-P(\bm{X}_i)\mathbb{E}\left[Y_i(0)\right]\right\}\\
=&\frac{1}{P\left[D=1\right] (1-P(\bm{X}_i))}\left\{(1-P(\bm{X}_i))P\left[D=1\right]\mathbb{E}\left[(Y_i(1)-Y_i(0))\big|D_i=1\right] \right.\\
& \left. + \mathbb{E}\left[Y_i(0)\right](\mathbb{E}\left[D_i\right]-P(\bm{X}_i))\right\}\\
=&\mathbb{E}\left[(Y_i(1)-Y_i(0))\big|D_i=1\right]
\end{align*}
The first equality is by definition, the second is after some simply algebra and taking into account that $Y_i=D_i Y_i(1)+(1-D_i) Y_i(0)$ and $D_i Y_i=D_i Y_i(1) + D_i Y_i(0) - D_i Y_i(0)$ given $D_i=\left\{0,1\right\}$. The third is an application of the law of iterative expectations  for the first term, and the Assumption \ref{assump2} (ignorability) for the second term. The fourth equality takes again into account that $D_i=\left\{0,1\right\}$, and $P\left[D_i=1\right]=P\left[D=1\right]$ (random sampling). Finally, we take into account that $\mathbb{E} \left[D_i|\bm{X}_i=\bm{x}_i\right]=P\left[D_i=1|\bm{X}_i=\bm{x}\right]=P(\bm{X}_i)$. Notice that Assumption \ref{assump3} is required to have $P\left[D=1\right](1-P(\bm{X}_i))\neq 0$.

\begin{remark}
Observe that $\mathbb{E}\left[Y_i(0)D_i\right]=P\left[D=1\right]\mathbb{E}\left[Y_i(0)|D_i=1\right]$. Therefore, we also obtain the required statement, if 

\begin{equation*}
    \mathbb{E}\left[Y_i(0)|D_i=1\right]=\mathbb{E}\left[Y_i(0)|D_i=0\right]=\mathbb{E}\left[Y_i(0)\right]
\end{equation*}
conditional on $\bm{X}_i=\bm{x}_i$ is satisfied, that is, the conditional mean assumption on control group is satisfied \citep[p.~1316]{Angrist1999}. Outcome of untreated individuals does not determine participation.
\end{remark}
\end{proof}

We require estimators for $P\left[D=1\right]$ and $P(\bm{X}_i)$ to estimate $\Delta_i^{T*}$. We use $\frac{1}{N}\sum_{i=1}^N D_i$, $N=N_0+N_1$, as estimator for the marginal probability of treatment, and a logit model to estimate the conditional probability of treatment. Obviously, these estimators can be changed. Notice that this introduces an extra source of variability as we do not observe $P\left[D=1\right]$ and $P(\bm{X}_i)$ in observational data sets. This kind of issue is discussed by \cite{Athey2015} who point out that transformations like equation \ref{eq4} does not optimally use all available information, this may mean extra variability. Additionally, \cite{Abadie2016} show that ATE using matching based on the estimated propensity score is more efficient than matching based on the true propensity score. On the other hand, ATET using the estimated propensity score can be more o less efficient compared to the true propensity score. We performed some simulation exercises that show better inferential performance based on the estimated propensity score (available upon authors request). In addition, as our ATET (ATE) is asymptotically biased corrected, the mean squared error decomposition shows that variance explains the most part (75\% on average), as expected. However, it remains a proportion to finite sample bias effect (25\% on average). 

Give that ${\Delta_i^{T*}}$ in equation \ref{eq3} is a conditional unbiased estimator for $\Delta^T$ (Proposition \ref{prop1}), we select $k$ in the kNN regression framework such that

\begin{equation}\label{eq4}
    \argminA_{{k}}\frac{1}{G}\sum_{g=1}^{G}\left\{ \left(\hat{\Delta}_g^{BC}-\frac{1}{N_g}\sum_{i=1}^{N_g}\widehat{\Delta_{i,g}^{T*}}\right)^2\right\},
\end{equation}

\noindent where 

\begin{equation*}
    \widehat{\Delta_{i,g}^{T*}}=\left(\frac{1}{N_g}\sum_{i=1}^{N_g}D_{i,g}^{Test}\right)^{-1}\left(\frac{Y_{i,g}^{Test}(D_{i,g}^{Test}-\widehat{P(\bm{X}_{i,g}^{Test})})}{1-\widehat{P(\bm{X}_{i,g}^{Test}})}\right),
\end{equation*}

\begin{equation*}
    \hat{\Delta}_g^{TM}=\frac{1}{N_{1,g}}\sum_{i\in\left\{D_{i,g}^{Test}=1\right\}}\left[Y_{i,g}^{Test}(1)-\frac{1}{K}\sum_{j\in A_K^i(\bm{x}_g^{Test})}Y_{j,g}^{Train}(0)\right],
\end{equation*}

\begin{equation*}
    \hat{B}_g^{TM}=\frac{1}{N_{1,g}}\sum_{i\in\left\{D_{i,g}^{Test}=1\right\}}\left[\hat{\mu}_0(\bm{x}_i)_g^{Test}-\frac{1}{K}\sum_{j\in A_K^i(\bm{x}_g^{Test})}\hat{\mu}_0(\bm{x}_j)_{g}^{Train}\right],
\end{equation*}

\noindent and

\begin{equation*}
   \hat{\Delta}_g^{BC}=\hat{\Delta}_g^{TM}-\hat{B}_g^{TM} 
\end{equation*}

\noindent where $G$ is the number of groups, $N_g$ is the sample size of each group, and $N_{1,g}$ is the sample size of the treated group in group $g$, $g=1,2,\dots,G$. Superscripts $Train$ and $Test$ refer to train and test data sets (see below).

Notice that $\widehat{\Delta_{i,g}^{T*}}$ is the sample version of the conditional unbiased estimator of the ATET $\Delta_i^{T*}$, so its average is a consistent estimator for ATET. On the other hand, $\hat{\Delta}_g^{BC}$ is bias-corrected version of the matching estimator \citep{Abadie2006}.

kNN regression is a standard approach in the machine and statistical learning communities \citep[p.~14]{Hastie2009}. So, we follow $k$-fold cross-validation to select the optimal $k$ \citep[p.~241]{Hastie2009}. In particular, we randomly split the data set in $G$ roughly equal-sized groups such that we keep the proportion between treated and untreated individuals in each group, then we calculate $\frac{1}{N_g}\sum_{i=1}^{N_g}\widehat{\Delta_{i,g}^{T*}}$, $g=1,2,\dots,G$. We fit the model using $G-1$ groups, that is, we fix $k$, and identify the $k$ nearest neighbors in the control group at the training data set in terms of the Euclidean norm to individual $i$ in the treated group at the $G$-th left group (test data set), that is,

\begin{equation*}
|A_k^i(\bm{x}_g^{Test})|=\sum_{j\in A_k^i(\bm{x}_g^{Test})}\mathbbm{1}\left\{||\bm{X}_{i,g}^{Test}-\bm{X}_j^{Train}||<||\bm{X}_{i,g}^{Test}-\bm{X}_l^{Train}||\right\}=k,    
\end{equation*}

\noindent then we estimate $\frac{1}{k}\sum_{j\in A_k^i(\bm{x}_g^{Test})}Y_{j,g}^{Train}(0)$, and obtain $\hat{\Delta}_{i,g}^{TM}=Y_{i,g}^{Test}(1)-\frac{1}{k}\sum_{j\in A_k^i(\bm{x}_g^{Test})}Y_{j,g}^{Train}(0)$ and $\hat{\Delta}_{i,g}^{BC}=\hat{\mu}_0(\bm{x}_i)_{g}^{Test}-\frac{1}{k}\sum_{j\in A_k^i(\bm{x}_g^{Test})}\hat{\mu}_0(\bm{x}_j)_{g}^{Test}$ for $i\in\left\{D_{i,g}^{Test}=1\right\}$. So, we can calculate the mean squared error for the $G$-th left group. This is done for $g=1,2,\dots,G$ and different values of $k$ such that we average over the $G$ prediction errors given $k$, and select the $k$ which has the minimum average error (see program \ref{eq4}).

Given an optimal $k=k^*$, we use equation \ref{eqBC_ATET} to estimate the average treatment effect on treated using the whole sample size. As we said, this estimator converges in distribution to a normal distribution, so we use the asymptotic result to build confidence intervals.

Algorithm \ref{Alg1} summarizes our methodological proposal:

\begin{algorithm}
	\caption{Optimal $k$ in kNN to estimate ATET}\label{Alg1}
	\begin{algorithmic}[1] 
	    \State Calculate $\frac{1}{N}\sum_{i=1}^N D_i$.
	    \For{\texttt{$i=1,2,\dots,N$}}
	        \begin{itemize}
	            \item  Fit a logit model where the dependent variable is $D_i$ to get $\widehat{P(\bm{X}_i)}$.
	            \item Calculate the $\widehat{\Delta_i^{T*}}$, the sample version of equation \ref{eq3}.
	        \end{itemize}
	    \EndFor
	    \State Split randomly the data set in $G$ roughly equal-sized groups such that the proportion between treated and untreated in each group is keep the same as in the original data set.
	    \State Let us set the training data set as $G-1$ groups, and the test data set the $g$-th left group.
	    \For{\texttt{$k=1,2,\dots,K$}}
	        \For{\texttt{$g=1,2,\dots,G$}}
	            \begin{itemize}
	                \item Calculate $\frac{1}{N_g}\sum_{i=1}^{N_g}\widehat{\Delta_{i,g}^{T*}}$, where $N_g$ is the sample size of $g$-th group (test data set). 
	                \item Find the $k$ nearest neighbors of $\bm{X}_{i,g}$ for treated individuals of the $g$th group ($D_{i,g}=1$ in the test data set) among the untreated individuals in the training data set ($D_{j,h\notin g}=0$ in the remaining $G-1$ groups).
	                \item Calculate $\frac{1}{k}\sum_{j\in A_k^i(\bm{x}_g^{Test})}Y_{j,g}^{Train}(0)$ and $\frac{1}{k}\sum_{j\in A_k^i(\bm{x}_g^{Test})}\hat{\mu}_0(\bm{x}_j)_{g}^{Train}$.
	                \item Obtain $\hat{\Delta}_{i,g}^{TM}=Y_{i,g}^{Test}(1)-\frac{1}{k}\sum_{j\in A_k^i(\bm{x}_g^{Test})}Y_{j,g}^{Train}(0)$ and $\hat{B}_{i,g}^{TM}=\hat{\mu}_0(\bm{x}_i)_{g}^{Test}-\frac{1}{k}\sum_{j\in A_k^i(\bm{x}_g^{Test})}\hat{\mu}_0(\bm{x}_j)_{g}^{Train}$ for $i\in\left\{D_{i,g}^{Test}=1\right\}$.
	                \item Estimate $\hat{\Delta}_g^{TM}=\frac{1}{N_{1,g}}\sum_{i\in\left\{D_{i,g}^{Test}=1\right\}}\hat{\Delta}_{i,g}^{TM}$, $\hat{B}_g^{TM}=\frac{1}{N_{1,g}}\sum_{i\in\left\{D_{i,g}^{Test}=1\right\}}\hat{B}_{i,g}^{TM}$ and $\hat{\Delta}_g^{BC}=\hat{\Delta}_g^{TM}-\hat{B}_g^{TM}$ 
	            \end{itemize}
	       \EndFor
	       \begin{itemize}
	           \item Calculate the mean squared error: $$\frac{1}{G}\sum_{g=1}^{G}\left\{ \left(\hat{\Delta}_g^{BC}-\frac{1}{N_g}\sum_{i=1}^{N_g}\widehat{\Delta_{i,g}^{T*}}\right)^2\right\}$$
	       \end{itemize}
	       \EndFor
	       \State Select $k$ that minimizes the mean squared error ($k^*$).
	       \For{\texttt{$i=1,2,\dots,N_1$}}
	        \begin{itemize}
	           \item Given $k^*$ and the whole data set, find the $k^*$ nearest neighbors in the untreated group to each individual in the treated group.
	           \item Calculate $\hat{\Delta}_i^{BC}=Y_i(1)-\frac{1}{k^*}\sum_{j\in A_{k^*}^i(\bm{x})}Y_j(0)$
	           \end{itemize}
	       \EndFor
	       \State Use $\hat{\Delta}^{BC}$ as a measure of central tendency of the ATET at neighborhood level.
	       \State Use theorems 4 and 7 in \cite{Abadie2006} to build confidence intervals of the ATET. 
	\end{algorithmic}
\end{algorithm}

\section{A Monte Carlos study}\label{Simulations}

We set exactly the same simulation setting as \cite{Otsu2017} who proposed a weighted bootstrap to perform inference of matching estimators for ATE and ATET as it is well known that naive bootstrap produces invalid inference due to failing to reproduce the distribution of the number of times each unit is used as a match \citep{Abadie2008}.

We show the simulation setting for exposition purposes:

\begin{align*}
    Y_i(1) = & \tau + m_j(||\bm X_i||)+\epsilon_i,\\
    Y_i(0) = & m_j(||\bm X_i||)+\epsilon_i,\\
    D_i = & \mathbbm{1}\left\{P(\bm X_i)\geq v_i\right\}, \ v_i\sim U\left[0,1\right],\\
    P(\bm X_i) = & \gamma_1+\gamma_2 ||\bm X_i||, \ X_i=\left[X_{i1},\dots,X_{iP}\right]^{\top},\\
    X_{ij}  = & \psi_i|\zeta_{ij}|/||\zeta_{i}||, \ j=1,2,\dots,P,\\
    & \psi_i\sim U\left[0,1\right], \ \zeta_i\sim N(\bm 0, \bm I_P), \ \epsilon_i\sim N(0,0.2^2),
\end{align*}

where $\epsilon_i$, $v_i$, $\psi_i$ and $\zeta_i$ are mutually independent. The ATET as well as ATE is $\tau = 0.5$,  $\gamma_1=0.15$, $\gamma=0.7$, $P=5$ and there are six different curves for $m_j(||\bm X_i||)$ (see \cite{Otsu2017} for details):

\begin{align*}
    m_1(||\bm X_i||) = & 0.15+0.7||\bm X_i||,\\
    m_2(||\bm X_i||) = & 0.1+0.5||\bm X_i||+0.5\exp \left(-200(||\bm X_i||-0.7)^2\right),\\
    m_3(||\bm X_i||) = & 0.8-2(||\bm X_i||-0.9)^2-5(||\bm X_i||-0.7)^3-10(||\bm X_i||-0.6)^{10},\\
    m_4(||\bm X_i||)= & 0.2+(1-||\bm X_i||)^{0.5}-0.6(0.9-||\bm X_i||)^2,\\
    m_5(||\bm X_i||)= & 0.2+(1-||\bm X_i||)^{0.5}-0.6(0.9-||\bm X_i||)^2-0.1||\bm X_i|| \cos (30||\bm X_i||),\\
    m_6(||\bm X_i||)= & 0.4+0.25\sin (8||\bm X_i||-5)+0.4\exp \left(-16(4||\bm X_i||-2.5)^2\right).\\
\end{align*}

In addition, we perform other simulation exercises where the outcome variable is binary. In particular, 

\begin{align*}
    Y_i(1) & = \begin{Bmatrix} 1, & \tau + \beta m_j(||\bm X_i||)+\epsilon_i>0\\
    0, & \tau + \beta m_j(||\bm X_i||)+ \epsilon_i\leq 0\end{Bmatrix},\\
    Y_i(0) & = \begin{Bmatrix} 1, & \beta m_j(||\bm X_i||)+\epsilon_i>0\\
    0, &  \beta m_j(||\bm X_i||)+\epsilon_i\leq 0\end{Bmatrix},
\end{align*}

where $\epsilon_i\sim LG(0,1)$, $LG$ denotes a logistic distribution, $\beta=0.5$, $\tau= 0.5$, and $P(D_i=1)=P(\gamma_1+\gamma_2 ||\bm X_i||>v_i)=F_{LG}(\gamma_1+\gamma_2 ||\bm X_i||)$, $\gamma_1=0.15$, $\gamma_2=0.4$, $v_i\sim LG(0, 1)$ and $F_{LG}$ is the distribution function. Other components of the specification are as in the previous setting.

Observe that in this setting $ATE_i=\mathbbm{E}(Y_i(1)-Y_i(0)|\bm X_i=\bm x_i, Z_i = z_i)=P(Y_i(1)=1|\bm X_i=\bm x_i, Z_i = z_i)-P(Y_i(0)=0|\bm X_i=\bm x_i, Z_i = z_i)=F_{LG}(\tau + \beta m_j(||\bm X_i||))-F_{LG}(\beta m_j(||\bm X_i||))$, which has the same analytical expression for the ATET, but conditional on the relevant sample ($D_i=1$).
    
Our results are based on a sample size equal to 100 \citep{Otsu2017}.\footnote{We perform other simulation exercises using larger sample sizes. Outcomes show similar results. Available upon authors request.} We have 1,000 replications for each simulation exercise.

First, we present the results for the average treatment effects on treated for the continuous outcome variable.\footnote{Results for the average treatment effects are in the Appendix, subsection \ref{simATE}. Outcomes are similar to ATET.} Table \ref{tab:MRSE1} displays the mean relative squared errors (MRSE),

\begin{equation*}
    MRSE = \frac{1}{S_k}\sum_{s=1}^{S_k} \left(\frac{\hat{\Delta}^{BC}_s(k)-ATET}{\hat{\Delta}^{BC}_s(k^*)-ATET}\right)^2, \ k^*\neq k, k=1,2,\dots 20,
\end{equation*}
where $\hat{\Delta}^{BC}_s(k)$ and $\hat{\Delta}^{BC}_s(k^*)$ are the ATET estimates using any $k$ and the optimal $k=k^*$, and $S_k$ is the number of times that $k$ was not optimal. 

We see in Table \ref{tab:MRSE1} that there are relative large  advantages using our optimal $k^*$ as all figures are larger than 1.
\begin{table}[]
\caption{Mean relative squared errors continuous outcome: Average treatment effects on treated}\label{tab:MRSE1}
\centering
\begin{tabular}{lllllll}
\hline
$k$  & $m_{1}$ & $m_{2}$ & $m_{3}$ & $m_{4}$ & $m_{5}$ & $m_{6}$ \\ \hline
1  & 479,945.8 & 361.6    & 71.9     & 74.9     & 64,319.7 & 5,308.1  \\
2  & 14,005.9  & 194.9    & 90.9     & 54.2     & 462.4    & 24,887.9 \\
3  & 658.6     & 126.7    & 129.6    & 10.4     & 18,965.7 & 7,391.0  \\
4  & 25,471.9  & 85.2     & 118.5    & 14.5     & 11.4     & 92.9     \\
5  & 805.4     & 43.1     & 121.7    & 11.0     & 2,366.4  & 101.1    \\
6  & 9,758.6   & 19.9     & 138.3    & 19.5     & 7,236.0  & 290.2    \\
7  & 15,678.7  & 20.7     & 142.0    & 21.0     & 29,098.9 & 171.5    \\
8  & 1,519.8   & 16.6     & 177.4    & 12.9     & 27,464.4 & 183.4    \\
9  & 29,276.6  & 14.4     & 193.9    & 10.2     & 23,308.7 & 336.1    \\
10 & 10,104.8  & 19.2     & 220.3    & 14.0     & 38,362.6 & 574.2    \\
11 & 1,616.9   & 19.8     & 249.1    & 15.1     & 41,302.2 & 709.1    \\
12 & 6,571.2   & 22.3     & 251.4    & 28.6     & 38,928.3 & 1,318.5  \\
13 & 23,269.0  & 25.5     & 336.5    & 47.6     & 40,001.6 & 1,474.0  \\
14 & 14,743.3  & 27.5     & 400.7    & 38.0     & 46,703.4 & 2,351.0  \\
15 & 10,691.7  & 29.2     & 395.8    & 55.7     & 42,344.9 & 3,035.4  \\
16 & 23,648.5  & 33.4     & 424.8    & 92.7     & 44,138.8 & 3,372.0  \\
17 & 19,488.5  & 41.6     & 459.2    & 126.6    & 45,173.9 & 4,027.6  \\
18 & 20,453.9  & 44.2     & 502.1    & 172.5    & 41,846.3 & 3,577.2  \\
19 & 26,292.6  & 44.3     & 486.7    & 214.3    & 50,218.6 & 3,886.7  \\
20 & 29,382.8  & 45.0     & 547.4    & 263.2    & 51,969.0 & 2,973.6   \\\hline
\multicolumn{7}{l}{\footnotesize $MRSE = \frac{1}{S_k}\sum_{s=1}^{S_k} \left(\frac{\hat{\Delta}^{BC}_s(k)-ATET}{\hat{\Delta}^{BC}_s(k^*)-ATET}\right)^2, \ k^*\neq k, k=1,2,\dots 20$ where}\\
\multicolumn{7}{l}{\footnotesize  $\hat{\Delta}^{BC}_s(k)$ and $\hat{\Delta}^{BC}_s(k^*)$ are the ATET estimates using any $k$ and}\\
\multicolumn{7}{l}{\footnotesize  the optimal $k=k^*$, and $S_k$ is the number of times that $k$ was not}\\
\multicolumn{7}{l}{\footnotesize optimal. It seems that there are relative large advantages using}\\
\multicolumn{7}{l}{\footnotesize our optimal $k^*$.}\\
\hline\hline
\end{tabular}
\end{table}

Figure \ref{fig:IWRCT} shows the average over 1,000 replications of the ratios between the length of the 95\% confidence intervals associated with any number of neighbors ($k$, $x$ axis) and the optimal selection using our proposal. This is done for the six data generating settings ($m_j(||\bm X_i||), j=1,2,\dots 6$). We see in this figure that a small $k$ implies relative wide intervals, whereas large $k$ implies narrow intervals. This is expected as small $k$ implies high variability compared with large $k$. Observe that intervals lengths decreases very fast, and stabilize around $k=10$. On the other hand, we see in Figure \ref{fig:TypeICT} that the type I error due to rejecting the null hypothesis $ATET=\tau$, given a significance level equal to 5\%, is the lowest using $k=1$. Using $k$ minimizing the mean squared error gives a relatively small type I error compared with other $k$ selection.

Figure \ref{fig:TypeICTF} displays the histograms of the optimal $k$ in our six simulation settings. We see that optimal selection is dominated by small $k$. This could be evidence to support $k=1$, which is widely used in many applications. Although, it seems that $k=\left\{2,3\right\}$ are the most relevant, and large values of $k$ cannot be discarded, especially when using large sample sizes (we performed other simulations exercises that show this). 

\begin{figure}
\centering
  \includegraphics[scale=0.5]{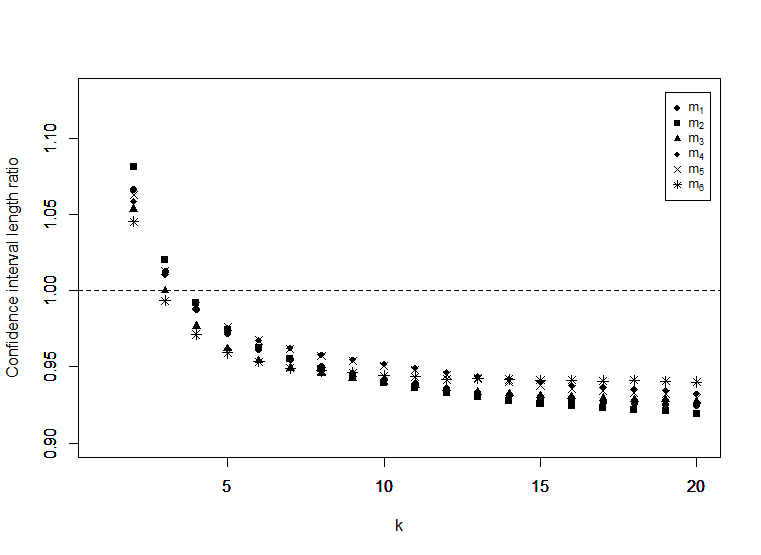}
\caption{95\% confidence interval length ratio continuous outcome: Average treatment effects on treated. Small $k$ implies relative wide intervals, whereas large $k$ implies narrow intervals. The dotted line is optimal selection of $k$ minimizing mean squared error. By construction has a ratio equal to 1.}
  \label{fig:IWRCT}
\end{figure}

\begin{figure}
\centering
  \includegraphics[scale=0.5]{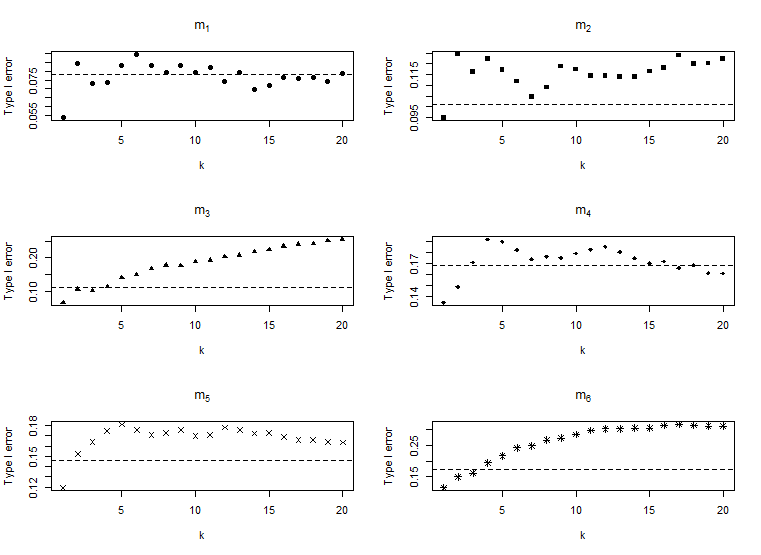}
\caption{Type I errors (5\% nominal size) continuous outcome: Average treatment effects on treated. Optimal $k$ gives a relatively small type I error compared with other $k$ selection. $k=1$ always gets the best nominal size.}
  \label{fig:TypeICT}
\end{figure}

\begin{figure}
\centering
  \includegraphics[scale=0.5]{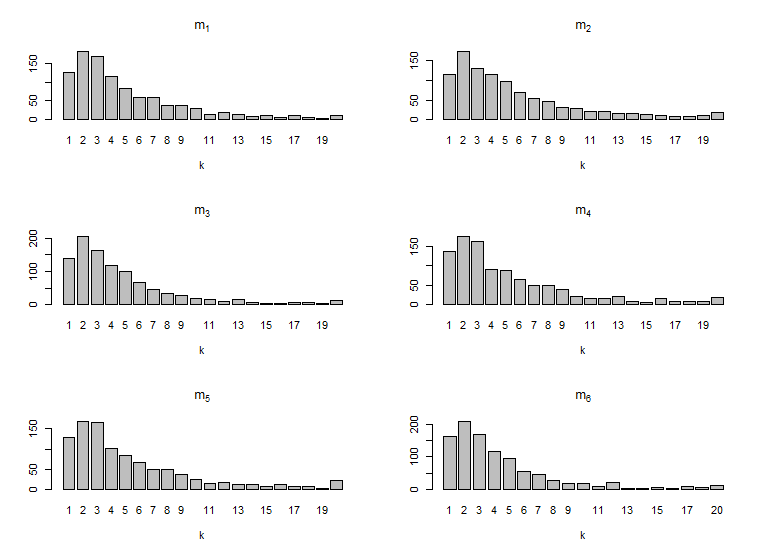}
\caption{ Optimal $k$ frequencies continuous outcome: Average treatment effects on treated. It seems that small $k$ dominates optimal $k$ choice. Although, large $k$ cannot be discarded.}
  \label{fig:TypeICTF}
\end{figure}

Results of the ATET for our simulation setting using a binary outcome are qualitatively similar to the continuous outcome. Table \ref{tab:MRSE2} show the mean relative squared errors. Again, all figures are larger than 1, which means that using a non-optimal $k$ implies larger estimation errors.

\begin{table}[]
\caption{Mean relative squared errors binary outcome: Average treatment effects on treated}\label{tab:MRSE2}
\centering
\begin{tabular}{lllllll}
$k$  & $m_{1}$ & $m_{2}$ & $m_{3}$ & $m_{4}$ & $m_{5}$ & $m_{6}$ \\ \hline
1  & 168.8    & 143.8    & 185.1    & 222.2    & 8,684.6  & 2,581.9  \\
2  & 103.1    & 95.7     & 72.8     & 52.1     & 51.5     & 1,342.0  \\
3  & 54.6     & 54.3     & 69.7     & 56.1     & 56.1     & 737.1    \\
4  & 56.4     & 21.3     & 23.9     & 17.9     & 114.1    & 441.2    \\
5  & 50.0     & 23.1     & 30.6     & 27.0     & 481.3    & 478.2    \\
6  & 37.2     & 13.7     & 15.6     & 20.9     & 1,460.7  & 466.2    \\
7  & 51.7     & 13.1     & 18.3     & 24.3     & 1,228.6  & 256.8    \\
8  & 18.3     & 18.5     & 26.5     & 22.3     & 1,917.7  & 230.9    \\
9  & 34.8     & 17.0     & 26.1     & 27.0     & 2,316.8  & 317.2    \\
10 & 86.7     & 26.4     & 34.6     & 36.3     & 2,720.1  & 463.3    \\
11 & 178.4    & 34.6     & 44.4     & 62.4     & 2,198.2  & 352.9    \\
12 & 188.5    & 36.4     & 39.2     & 54.6     & 2,082.1  & 314.2    \\
13 & 246.8    & 36.0     & 38.8     & 60.4     & 3,024.4  & 363.5    \\
14 & 327.1    & 41.5     & 42.7     & 68.9     & 3,977.1  & 517.1    \\
15 & 383.3    & 47.5     & 47.5     & 88.9     & 5,120.8  & 709.2    \\
16 & 316.1    & 50.2     & 51.8     & 92.2     & 7,136.7  & 888.7    \\
17 & 292.2    & 51.1     & 46.4     & 98.8     & 8,163.2  & 1,222.5  \\
18 & 274.7    & 54.6     & 54.4     & 110.8    & 8,501.3  & 1,175.7  \\
19 & 254.6    & 55.6     & 60.3     & 110.4    & 8,730.9  & 1,105.3  \\
20 & 248.2    & 61.1     & 64.3     & 124.8    & 9,135.4  & 1,050.8 \\ \hline

\multicolumn{7}{l}{\footnotesize $MRSE = \frac{1}{S_k}\sum_{s=1}^{S_k} \left(\frac{\hat{\Delta}^{BC}_s(k)-ATET}{\hat{\Delta}^{BC}_s(k^*)-ATET}\right)^2, \ k^*\neq k, k=1,2,\dots 20$}\\
\multicolumn{7}{l}{\footnotesize  where $\hat{\Delta}^{BC}_s(k)$ and $\hat{\Delta}^{BC}_s(k^*)$ are the ATET estimates using}\\
\multicolumn{7}{l}{\footnotesize  any $k$ and the optimal $k=k^*$, and $S_k$ is the number of times}\\
\multicolumn{7}{l}{\footnotesize that $k$ was not optimal. It seems that there are relative}\\
\multicolumn{7}{l}{\footnotesize large advantages using an optimal $k$.}\\
\hline\hline
\end{tabular}
\end{table}

Figure \ref{fig:IWRCTbin} shows the relative 95\% confidence intervals length ratios associated with $k$ and $k^*$. All six simulation settings look very similar with wider confidence intervals when $k$ is small, and narrower intervals when $k$ is large. There is stabilization around $k=10$. Figure \ref{fig:TypeICTbin} shows type I errors using 5\% as significance level. Larger $k$ implies larger errors, and using $k^*$ implies small type I errors in general, except when $k$ is very small. Small type I error can be explained by means of Figure \ref{fig:TypeICTbinF} where we see that optimal selection of $k$ is also based on small $k$. Here $k^*=\left\{2,3,4,5\right\}$ are particularly relevant. 

In general, we can observe that our proposal implies relatively smaller estimation errors achieving a compromise between confidence interval length and type I error having as a result of small optimal $k$ most of the times.   

\begin{figure}
\centering
  \includegraphics[scale=0.5]{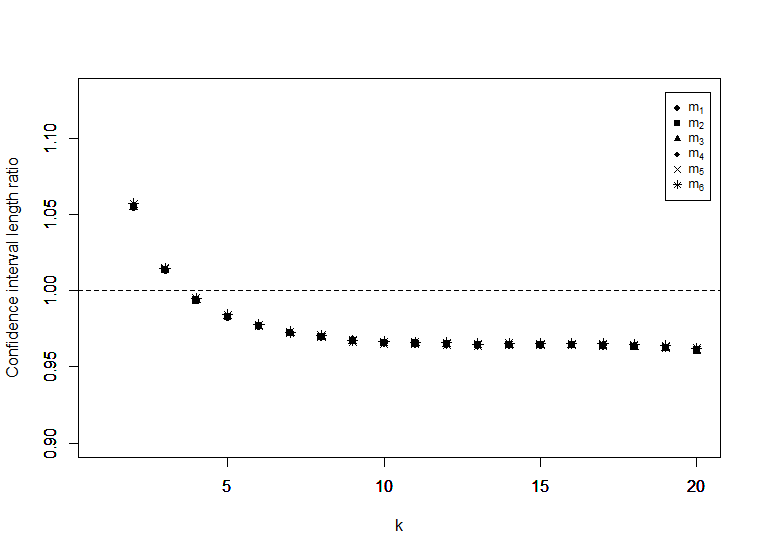}
\caption{95\% confidence interval length ratio binary outcome: Average treatment effects on treated. Small $k$ implies relative wide intervals, whereas large $k$ implies narrow intervals. The dotted line is optimal selection of $k$ minimizing mean squared error. By construction has a ratio equal to 1.}
  \label{fig:IWRCTbin}
\end{figure}

\begin{figure}
\centering
  \includegraphics[scale=0.5]{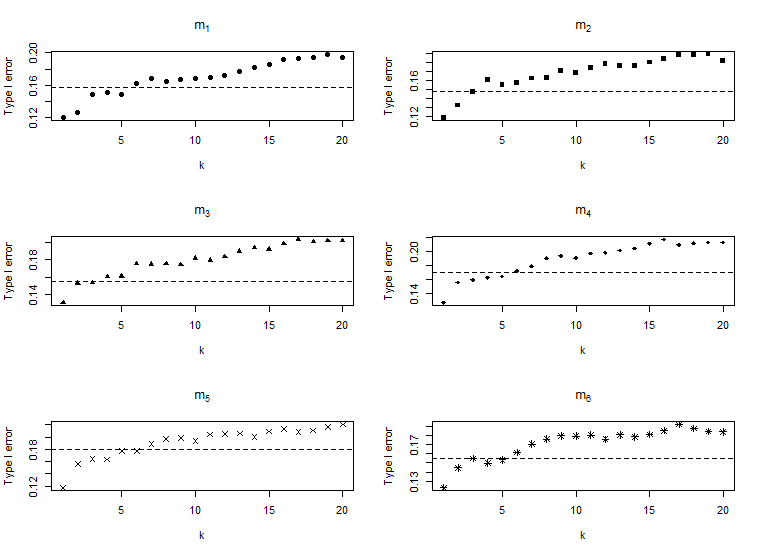}
\caption{Type I errors (5\% nominal size) binary outcome: Average treatment effects on treated. Optimal $k$ gives a relatively small type I error compared with other $k$ selection. $k=1$ always gets the best nominal size.}
  \label{fig:TypeICTbin}
\end{figure}

\begin{figure}
\centering
  \includegraphics[scale=0.5]{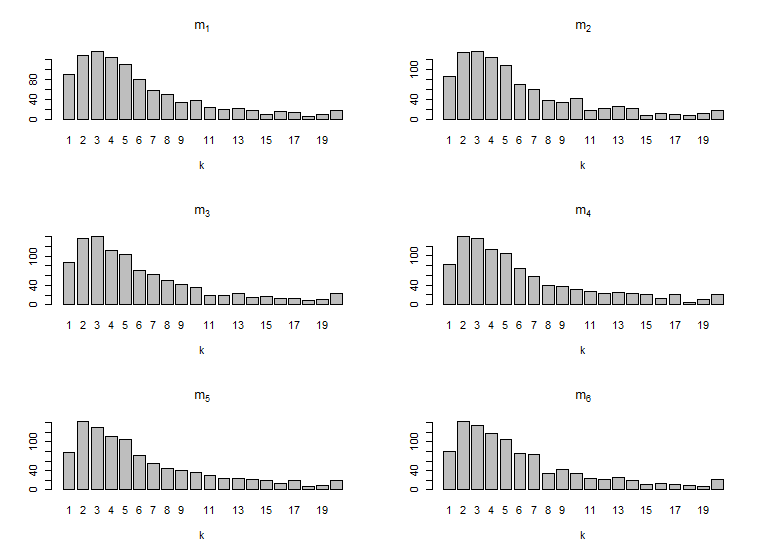}
\caption{ Optimal $k$ frequencies binary outcome: Average treatment effects on treated.  It seems that small $k$ dominates optimal $k$ choice. Although, large $k$ cannot be discarded.}
  \label{fig:TypeICTbinF}
\end{figure}

\section{Empirical example}\label{Application}

To illustrate our proposal we estimate the average treatment effect on treated of 401(k) participation on accumulated net financial assets. 401(k) is a retirement plan where savings contributions are provided by employers deduced from employees' payment before taxes such that taxes on capital gains are avoided. Employees should achieve some eligibility criteria such that the plan does not have universal applicability in principle. The impact evaluation of 401(k) on net financial assets has been done previously by \cite{Benjamin2003,chernozhukov2004effects, Conley2012,chernozhukov2017double} among others. We overcome the lack of random assignment following \cite{chernozhukov2017double} arguments. Enrollment in a 401(k) plan is exogenous when controlling for observable variables which drive employment decisions, such as income, when the plan initially became available. Therefore, our dataset is from the 1991 Survey of Income and Program Participation composed by 9,915 observations.

We control for the same set of variables as \cite{chernozhukov2004effects}. In particular, age (linear and squared), education (years), household size, income (seven levels), and binary variables indicating: defined benefit pension, home ownership, marital status (married), participation in IRA (individual retirement account) and two-earners status (both household heads contribute to household income). Details can be found in \cite{Benjamin2003}.  

Table \ref{tab:desc} shows unconditional mean difference tests of our outcome variables: accumulated net financial assets, and a binary variable indicating positive accumulated net financial assets. Employees enrolled in 401(k) plans have \$27.4 more accumulated net financial assets on average than employees who are not enrolled. The unconditional probability of positive net financial assets is 31\% on average higher for the former. We can see that there are unconditional statistically significant differences.

\begin{table}[]\caption{Descriptive statistics: Outcome variables}
\begin{tabular}{@{}llll@{}}
\hline
Response                                     & Mean treated  & Mean untreated & Statistic \\ \hline
Continuous & \$38.26K (79.09) & \$10.89K (55.26) & 16.28$^{**}$ \\
Binary  & 0.86 (0.35)         & 0.55 (0.5)          & 35.27 8$^{**}$      \\ \hline

\multicolumn{4}{l}{{\footnotesize Standard error in parenthesis. The test of the difference of two proportions uses a normally}}\\ 
\multicolumn{4}{l}{{\footnotesize  distributed test statistic calculated as $z=\frac{\hat{p}_1-\hat{p}_2}{(\hat{p}_p(1-\hat{p}_p)(1/n_1+1/n2))^{0.5}}$, $\hat{p}_p=\frac{x_1+x_2}{n_1+n_2}$, $x_l$ and $n_l$ are}}\\
\multicolumn{4}{l}{{\footnotesize  the number of successes and sample sizes in each group, $l=1,2$. The test of the difference}}\\
\multicolumn{4}{l}{{\footnotesize of two means uses a Student's t calculated as $t=\frac{\bar{x}_1-\bar{x}_2}{(s_{x_1}^2/n_1+s_{x_2}^2/n_2)^{0.5}}$.}}\\
\multicolumn{4}{l}{{\footnotesize $^1$ Null hypothesis: means are equal. Critical value at 5\% level of significance is 1.96}}\\
\multicolumn{4}{l}{{\footnotesize $^{**}$ Rejection of null hypothesis at 5\% significance level.}}\\

\hline
\hline

\end{tabular}\label{tab:desc}
\end{table}

Figure \ref{fig:visit} shows the ATET in the accumulated net financial assets using our optimal selection, $k^*=19$, and $k=1$, as a reference. We see similar ATET point estimates, which are approximately \$ 15K (similar outcomes are found by \cite{Conley2012} under the exogeneity assumption using a Bayesian approach for a local average treatment effects, LATE). However, the 95\% confidence interval using $k^*$ is (\$10.9K,\$18.9K), whereas $k=1$ interval is (\$9.1K,\$19.9K), that is 35\% larger. Results regarding the probability of positive accumulated net financial assets can be seen in Figure \ref{fig:emp}. Both, $k^*=19$ and $k=1$ indicate statistically significant ATET whose point estimate is approximately 19.9\% and 18.7\% respectively. The 95\% confidence intervals are (16.3\%, 21.3\%) and (17.7\%, 21.3\%) for $k=1$ and $k^*=19$. This means that using $k=1$ implies a 38.9\% wider interval than using $k^*=19$. Observe that both exercises indicate the same optimal $k$. This means that diagnostic balance tests are based on the same control groups. Figure \ref{fig:loveV} shows the love plot indicating that control groups based on $k^*=19$ satisfy mean balancing conditions.      

\begin{figure}
\centering
  \includegraphics[scale=0.5]{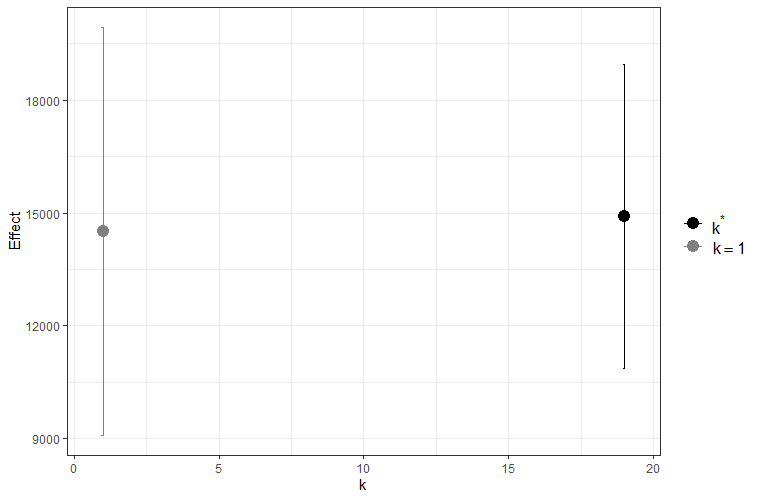}
\caption{Average treatment effects on treated due to 401(k) participation on accumulated net financial assets: Point estimates with $k=1$ (gray dot) and optimal $k^*$ (black dot), and 95\% confidence intervals. Both, $k=1$ and $k^*$ indicate 5\% significant positive ATET of approximately \$15K. 95\% confidence intervals using $k=1$ are 28\% wider compared to $k^*$.}
  \label{fig:visit}
\end{figure}

\begin{figure}
\centering
  \includegraphics[scale=0.5]{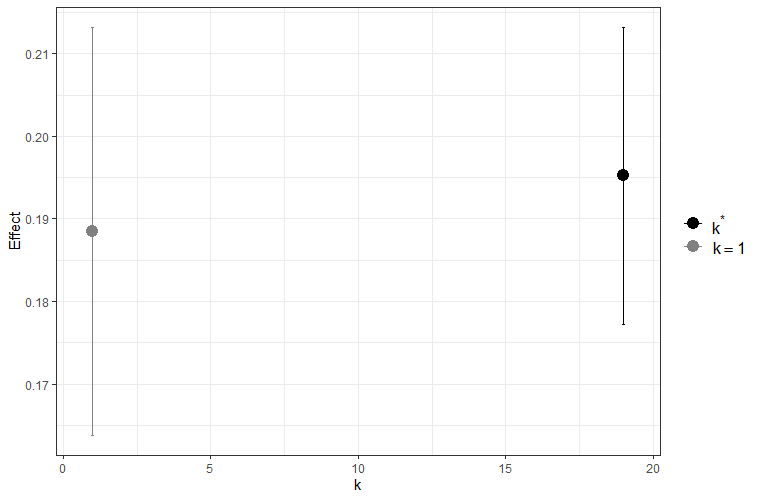}
\caption{Average treatment effects on treated due to 401(k) participation on the probability of positive accumulated net financial assets: Point estimates with $k=1$ (gray dot) and optimal $k^*$ (black dot), and 95\% confidence intervals. Both, $k=1$ and $k^*$ indicate 5\% significant positive ATET of approximately 19\%. 95\% confidence intervals using $k=1$ are 28\% wider compared to $k^*$.}
  \label{fig:emp}
\end{figure}

\begin{figure}[h]
\centering
  \includegraphics[scale=0.5]{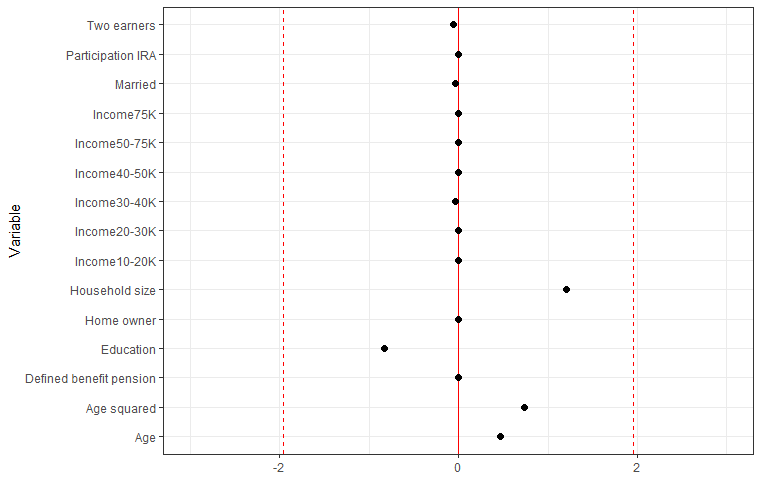}
\caption{Mean balancing condition for control variables used calculating ATET: Love plot using $k^*=19$. It seems that mean balancing conditions are satisfied as statistical tests (black dots) are inside critical levels (dashed red lines).}
  \label{fig:loveV}
\end{figure}

\section{Conclusions and final remarks}\label{Conclusions}
We present a simple approach to obtain an optimal selection of the number of control units ($k$) to build the ``synthetic'' or ``counterfactual'' unit for each individual in the k nearest neighbors algorithm. Our approach is based on a simple unbiased reference estimator (\textit{individual treatment effect on treated} for ATET and \textit{individual treatment effect} for ATE) such that we select $k$ minimizing the mean squared error with respect to the reference estimator. Therefore, the optimal $k=k^*$ achieves a balance between confidence interval length and type I error. 

Our application suggests that $k^*$ can be relatively large compared to the common practice of using $k=1$. This implies significant reductions regarding confidence interval lengths. Although, both choices show 5\% statistically significant effects of 401(k) enrollment in the probability of positive accumulated net financial assets, and its amount.  

A final remark is that using as reference estimator for treatments effects the one obtained using $k=1$ is an interesting alternative. We obtained same qualitative results, that is, a balance between interval length and type I error, and a considerable reduction in MRSE. Simulations exercises confirm this.

\clearpage
\bibliography{myReferences} 

\begin{thebibliography}{}

\bibitem[Abadie and Imbens, 2006]{Abadie2006}
Abadie, A. and Imbens, G. (2006).
\newblock Large sample properties of matching estimators for average treatment
  effects.
\newblock {\em Econometrica}, 74(1):235--267.

\bibitem[Abadie and Imbens, 2008]{Abadie2008}
Abadie, A. and Imbens, G. (2008).
\newblock On the failure of the bootstrap for matching estimators.
\newblock {\em Econometrica}, 76(6):1537--–1557.

\bibitem[Abadie and Imbens, 2011]{Abadie2011}
Abadie, A. and Imbens, G. (2011).
\newblock Bias-corrected matching estimators for average treatment effects.
\newblock {\em Journal of Business \& Economic Statistics}, 29(1):1--11.

\bibitem[Abadie and Imbens, 2016]{Abadie2016}
Abadie, A. and Imbens, G. (2016).
\newblock Matching on the estimated propensity score.
\newblock {\em Econometrica}, 84(2):781–--807.

\bibitem[Angrist and Krueger, 1999]{Angrist1999}
Angrist, J. and Krueger, A. (1999).
\newblock Empirical strategies in labor economics.
\newblock In Ashenfelter, O. and Card, D., editors, {\em Handbook of Labor
  Economics}, volume~3. ELSEVIER.

\bibitem[Athey et~al., 2015]{Athey2015}
Athey, S., Imbens, G., and Ramachandra, V. (2015).
\newblock Machine learning methods for estimating heterogeneous causal effects.
\newblock Technical report, Stanford University.

\bibitem[Benjamin, 2003]{Benjamin2003}
Benjamin, D. (2003).
\newblock Does 401(k) eligibility increase savings? {E}vidence from propensity
  score subclassification.
\newblock {\em Journal of Public Economics}, 87:1259--1290.

\bibitem[Cameron and Trivedi, 2005]{Cameron05}
Cameron, C. and Trivedi, K. (2005).
\newblock {\em Microeconometrics: Methods and Applications}.
\newblock Cambridge University Press.

\bibitem[Chernozhukov et~al., 2017]{chernozhukov2017double}
Chernozhukov, V., Chetverikov, D., Demirer, M., Duflo, E., Hansen, C., Newey,
  W., and Robins, J. (2017).
\newblock Double/debiased machine learning for treatment and causal parameters.
\newblock {\em The Econometrics Journal}, 21:1--68.

\bibitem[Chernozhukov and Hansen, 2004]{chernozhukov2004effects}
Chernozhukov, V. and Hansen, C. (2004).
\newblock The effects of 401 (k) participation on the wealth distribution: an
  instrumental quantile regression analysis.
\newblock {\em Review of Economics and statistics}, 86(3):735--751.

\bibitem[Conley et~al., 2012]{Conley2012}
Conley, T., Hansen, C., and Rossi, P. (2012).
\newblock Plausibly exogenous.
\newblock {\em The Review of Economics and Statistics}, 94(1):260--272.

\bibitem[Hastie et~al., 2009]{Hastie2009}
Hastie, T., Tibshirani, R., and Friedman, J. (2009).
\newblock {\em The Elements of Statistical Learning Data Mining, Inference, and
  Prediction}.
\newblock Springer, second edition edition.

\bibitem[Otsu and Rai, 2017]{Otsu2017}
Otsu, T. and Rai, Y. (2017).
\newblock Bootstrap inference of matching estimators for average treatment
  effects.
\newblock {\em Journal of the American Statistical Association},
  112(520):1720--1732.

\bibitem[Rosenbaum and Rubin, 1983]{Rosenbaum1983}
Rosenbaum, P.~R. and Rubin, D.~B. (1983).
\newblock The central role of the propensity score in observational studies for
  causal effects.
\newblock {\em Biometrika}, 70(1):41--55.

\bibitem[Rubin, 1978]{Rubin1978}
Rubin, D. (1978).
\newblock Bayesian inference for causal effects: The role of randomization.
\newblock {\em Annals of Statistics}, 6:34--58.

\end{thebibliography}
\bibliographystyle{apalike} 

\clearpage
\section{Appendix}\label{Appendx}
\subsection{Optimal selection of $k$ in kNN: Average treatment effects}\label{ref:ATE}

The average treatment effect is

\begin{equation*}
\Delta^{ATE}=\mathbb{E}\left[(Y_i(1)-Y_i(0))|\bm{X}_i=\bm{x}_i\right].  
\end{equation*}

A matching estimator for the ATE is given by

\begin{align*}
    \hat{\Delta}^{ATE}=\frac{1}{N}\sum_{i=1}^N \left[\hat{Y}_i(1)-\hat{Y}_i(0)\right],
\end{align*}
where
\begin{align*}
    \hat{Y}_i(1)&=\begin{Bmatrix} \frac{1}{k}\sum_{j\in A_k^i(\bm x)} Y_j(0), & D_i=0\\
    Y_i, & D_i=1\end{Bmatrix},
\end{align*}

\begin{align*}
    \hat{Y}_i(0)&=\begin{Bmatrix} Y_i, & D_i=0\\
    \frac{1}{k}\sum_{j\in A_k^i(\bm x)} Y_j(1), & D_i=1\end{Bmatrix},
\end{align*}

$A_k^i(\bm{x})$ defines the relevant control group conditional on being treated or untreated.

The bias of the ATE is 

\begin{equation*}
    B^{ATE}=\frac{1}{N}\left\{\sum_{i\in\left\{D_i=1\right\}}\left[\mu_0(\bm{x}_i)-\frac{1}{k}\sum_{j\in A_k^i(\bm x)}\mu_0(\bm{x}_j)\right] - \sum_{i\in\left\{D_i=0\right\}}\left[\mu_1(\bm{x}_i)-\frac{1}{k}\sum_{j\in A_k^i(\bm x)}\mu_1(\bm{x}_j)\right]\right\},
\end{equation*}

\noindent where $\mu_0(\bm{x}_i)=\mathbb{E}\left[Y|\bm{X}=\bm{x},D=0\right]$ and $\mu_1(\bm{x}_i)=\mathbb{E}\left[Y|\bm{X}=\bm{x},D=1\right]$. These can be estimated non-parametrically using a series expansion estimator to obtain $\hat{B}^{TM}$ \citep{Abadie2011}.

Therefore, the bias-corrected matching estimator is

\begin{equation*}
    \hat{\Delta}^{BCATE}=\hat{\Delta}^{ATE}-\hat{B}^{ATE}.
\end{equation*}

Theorem 4(i) in \cite{Abadie2006} shows asymptotic distribution convergence results for the bias corrected average treatment effect, and section 3.2 in their paper has required variance expressions (marginal and conditional).

The \textit{individual treatment effect} is 

\begin{equation}\label{eq3}
    {\Delta_i^{ATE*}}=\frac{Y_i(D_i-P(\bm{X}_i))}{P(\bm{X}_i)(1-P(\bm{X}_i))},
\end{equation}

such that $\mathbb{E}\left[{\Delta_i^{ATE*}}\big| \bm{X}_i=\bm{x}\right]=\Delta^{ATE}$ \citep{Athey2015}.

Therefore, the optimal selection of $k$ is given by solving the program

\begin{equation*}
    \argminA_{{k}}\frac{1}{G}\sum_{g=1}^{G}\left\{ \left(\hat{\Delta}_g^{BCATE}-\frac{1}{N_g}\sum_{i=1}^{N_g}\widehat{\Delta_{i,g}^{ATE*}}\right)^2\right\},
\end{equation*}

\noindent where 

\begin{equation*}
    \widehat{\Delta_{i,g}^{ATE*}}=\left(\frac{Y_{i,g}^{Test}(D_{i,g}^{Test}-\widehat{P(\bm{X}_{i,g}^{Test})})}{\widehat{P(\bm{X}_{i,g}^{Test})}(1-\widehat{P(\bm{X}_{i,g}^{Test})})})\right),
\end{equation*}

\begin{align*}
    \hat{\Delta}_g^{ATE}=&\frac{1}{N_{g}}\left\{\sum_{i\in\left\{D_{i,g}^{Test}=1\right\}}\left[Y_{i,g}^{Test}(1)-\frac{1}{k}\sum_{j\in A_k^i(\bm{x}_g^{Test})}Y_{j,g}^{Train}(0)\right]\right.\\
    &\left.-\sum_{i\in\left\{D_{i,g}^{Test}=0\right\}}\left[Y_{i,g}^{Test}(0)-\frac{1}{k}\sum_{j\in A_k^i(\bm{x}_g^{Test})}Y_{j,g}^{Train}(1)\right]\right\},
\end{align*}

\begin{align*}
    \hat{B}_g^{ATE}=&\frac{1}{N_{g}}\left\{\sum_{i\in\left\{D_{i,g}^{Test}=1\right\}}\left[\hat{\mu}_0(\bm{x}_i)_g^{Test}-\frac{1}{k}\sum_{j\in A_k^i(\bm{x}_g^{Test})}\hat{\mu}_0(\bm{x}_j)_{g}^{Train}\right]\right.\\
    &-\left.\sum_{i\in\left\{D_{i,g}^{Test}=0\right\}}\left[\hat{\mu}_1(\bm{x}_i)_g^{Test}-\frac{1}{k}\sum_{j\in A_k^i(\bm{x}_g^{Test})}\hat{\mu}_1(\bm{x}_j)_{g}^{Train}\right]\right\},
\end{align*}

\noindent and

\begin{equation*}
   \hat{\Delta}_g^{BCATE}=\hat{\Delta}_g^{ATE}-\hat{B}_g^{ATE}. 
\end{equation*}

An algorithm similar to Algorithm \ref{Alg1} can be implemented to obtain an optimal $k$ for the ATE given these definitions.

\subsection{Simulation outcomes: Average treatment effects}\label{simATE}

\begin{table}[h]
\caption{Mean relative squared errors continuous outcome: Average treatment effects}\label{tab:MRSE1ATE}
\centering
\begin{tabular}{lllllll}
\hline
$k$  & $m_{1}$ & $m_{2}$ & $m_{3}$ & $m_{4}$ & $m_{5}$ & $m_{6}$ \\ \hline
1  & 709.3    & 4,511.1  & 42.0     & 790.9    & 586.1    & 249.2    \\
2  & 12.8     & 31.1     & 97.7     & 403.6    & 200.8    & 63.2     \\
3  & 9.3      & 162.6    & 104.7    & 329.2    & 272.2    & 85.2     \\
4  & 20.4     & 805.9    & 26.0     & 17.5     & 54.9     & 49.1     \\
5  & 7.5      & 569.2    & 19.5     & 34.7     & 87.3     & 41.1     \\
6  & 11.4     & 382.9    & 25.5     & 39.5     & 539.8    & 27.3     \\
7  & 14.2     & 522.7    & 29.1     & 100.4    & 949.7    & 41.5     \\
8  & 11.8     & 855.1    & 31.7     & 171.5    & 1,055.5  & 46.0     \\
9  & 7.6      & 495.1    & 34.4     & 563.4    & 1,104.3  & 37.2     \\
10 & 9.8      & 697.9    & 40.9     & 856.8    & 946.1    & 55.8     \\
11 & 8.4      & 483.9    & 47.7     & 888.2    & 949.8    & 82.0     \\
12 & 12.8     & 372.7    & 52.9     & 1,050.7  & 1,073.7  & 71.3     \\
13 & 12.4     & 298.0    & 49.5     & 1,557.7  & 1,240.0  & 65.3     \\
14 & 10.7     & 360.9    & 56.7     & 1,759.5  & 1,445.1  & 49.8     \\
15 & 9.1      & 371.0    & 67.7     & 1,818.8  & 1,742.4  & 37.8     \\
16 & 14.7     & 429.1    & 68.6     & 1,756.0  & 1,915.3  & 38.3     \\
17 & 17.8     & 445.8    & 90.8     & 1,746.5  & 2,096.7  & 57.8     \\
18 & 10.9     & 382.9    & 104.9    & 1,782.8  & 2,252.2  & 84.7     \\
19 & 9.2      & 386.4    & 146.7    & 1,623.2  & 2,586.7  & 155.0    \\
20 & 9.4      & 233.2    & 184.2    & 1,739.2  & 2,722.4  & 212.9 \\ \hline
\multicolumn{7}{l}{\footnotesize $MRSE = \frac{1}{S_k}\sum_{s=1}^{S_k} \left(\frac{\hat{\Delta}^{BC}_s(k)-ATE}{\hat{\Delta}^{BC}_s(k^*)-ATE}\right)^2, \ k^*\neq k, k=1,2,\dots 20$}\\
\multicolumn{7}{l}{\footnotesize  where $\hat{\Delta}^{BC}_s(k)$ and $\hat{\Delta}^{BC}_s(k^*)$ are the ATE estimates using any $k$}\\
\multicolumn{7}{l}{\footnotesize  and the optimal $k=k^*$, and $S_k$ is the number of times that $k$ was}\\
\multicolumn{7}{l}{\footnotesize not optimal. It seems that there are relative large advantages}\\
\multicolumn{7}{l}{\footnotesize using our optimal $k^*$.}\\
\hline\hline
\end{tabular}
\end{table}

\begin{figure}
\centering
  \includegraphics[scale=0.5]{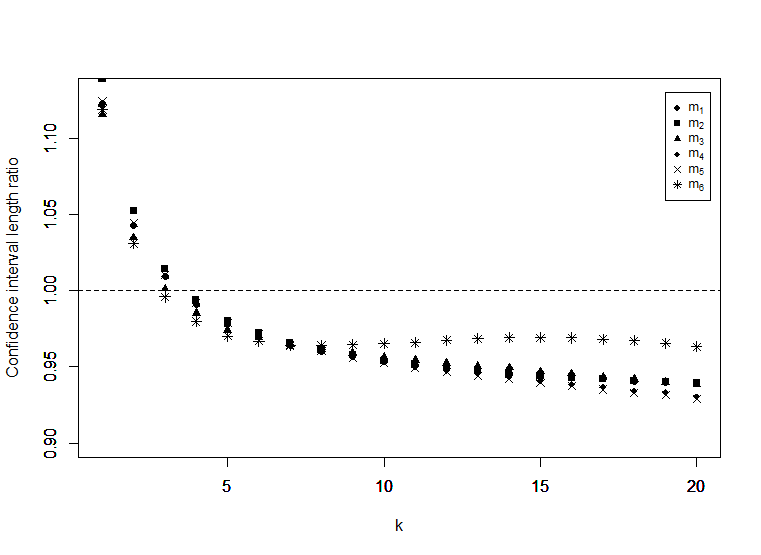}
\caption{95\% confidence interval length ratio continuous outcome: Average treatment effects. Small $k$ implies relative wide intervals, whereas large $k$ implies narrow intervals. The dotted line is the optimal selection of $k$ minimizing mean squared error. By construction has a ratio equal to 1.}
  \label{fig:IWRCATE}
\end{figure}

\begin{figure}
\centering
  \includegraphics[scale=0.5]{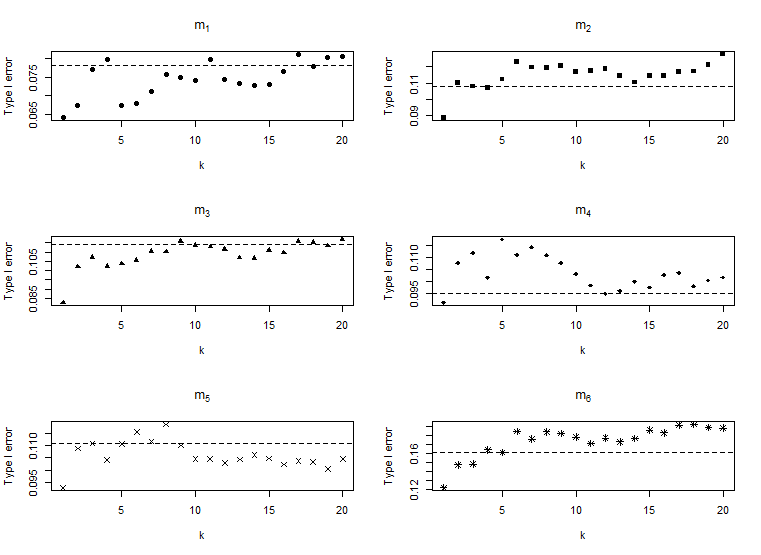}
\caption{Type I errors continuous outcome: Average treatment effects. The optimal $k$ gives an average type I error compared with other $k$ selection. $k=1$ always gets the best nominal size}
  \label{fig:IWRCATE}
\end{figure}

\begin{table}[h]
\caption{Mean relative squared errors binary outcome: Average treatment effects}\label{tab:MRSE2ATE}
\centering
\begin{tabular}{lllllll}
\hline
$k$  & $m_{1}$ & $m_{2}$ & $m_{3}$ & $m_{4}$ & $m_{5}$ & $m_{6}$ \\ \hline
1  & 73.5     & 123.3    & 191.2    & 275.5    & 194.6    & 1,444.2   \\
2  & 14.5     & 43.5     & 838.5    & 209.1    & 25.3     & 187,820.3 \\
3  & 12.7     & 41.6     & 906.2    & 266.3    & 42.3     & 378,910.0 \\
4  & 35.3     & 37.7     & 27.5     & 433.5    & 106.5    & 508,679.5 \\
5  & 15.7     & 28.7     & 26.1     & 308.9    & 149.3    & 519,137.8 \\
6  & 12.5     & 30.2     & 22.5     & 511.0    & 150.7    & 626,520.8 \\
7  & 12.8     & 49.5     & 38.6     & 754.6    & 148.2    & 739,373.9 \\
8  & 9.4      & 81.3     & 40.0     & 1,114.0  & 182.0    & 826,213.6 \\
9  & 19.0     & 59.8     & 31.3     & 1,048.2  & 133.6    & 755,273.2 \\
10 & 12.7     & 79.5     & 23.2     & 1,199.5  & 134.1    & 764,009.2 \\
11 & 13.9     & 44.8     & 26.8     & 988.9    & 140.7    & 786,963.7 \\
12 & 16.5     & 44.5     & 30.9     & 1,140.5  & 131.7    & 715,959.8 \\
13 & 24.7     & 33.5     & 18.0     & 1,375.7  & 139.4    & 727,003.1 \\
14 & 18.5     & 25.1     & 25.9     & 1,362.9  & 157.4    & 697,692.1 \\
15 & 23.3     & 20.4     & 48.2     & 1,213.5  & 194.6    & 711,163.3 \\
16 & 26.3     & 15.1     & 13.0     & 1,198.5  & 205.3    & 738,156.2 \\
17 & 23.8     & 15.0     & 14.2     & 1,269.1  & 196.7    & 695,079.4 \\
18 & 17.9     & 14.6     & 11.9     & 1,325.0  & 220.8    & 661,129.6 \\
19 & 18.2     & 14.7     & 12.0     & 1,267.1  & 235.7    & 615,817.0 \\
20 & 19.4     & 16.2     & 24.4     & 1,381.8  & 254.8    & 647,296.8 \\ \hline
\multicolumn{7}{l}{\footnotesize $MRSE = \frac{1}{S_k}\sum_{s=1}^{S_k} \left(\frac{\hat{\Delta}^{BC}_s(k)-ATE}{\hat{\Delta}^{BC}_s(k^*)-ATE}\right)^2, \ k^*\neq k, k=1,2,\dots 20$}\\
\multicolumn{7}{l}{\footnotesize  where $\hat{\Delta}^{BC}_s(k)$ and $\hat{\Delta}^{BC}_s(k^*)$ are the ATE estimates using}\\
\multicolumn{7}{l}{\footnotesize  any $k$ and the optimal $k=k^*$, and $S_k$ is the number of times}\\
\multicolumn{7}{l}{\footnotesize that $k$ was not optimal. It seems that there are relative}\\
\multicolumn{7}{l}{\footnotesize large advantages using an optimal $k$.}\\
\hline\hline
\end{tabular}
\end{table}

\begin{figure}
\centering
  \includegraphics[scale=0.5]{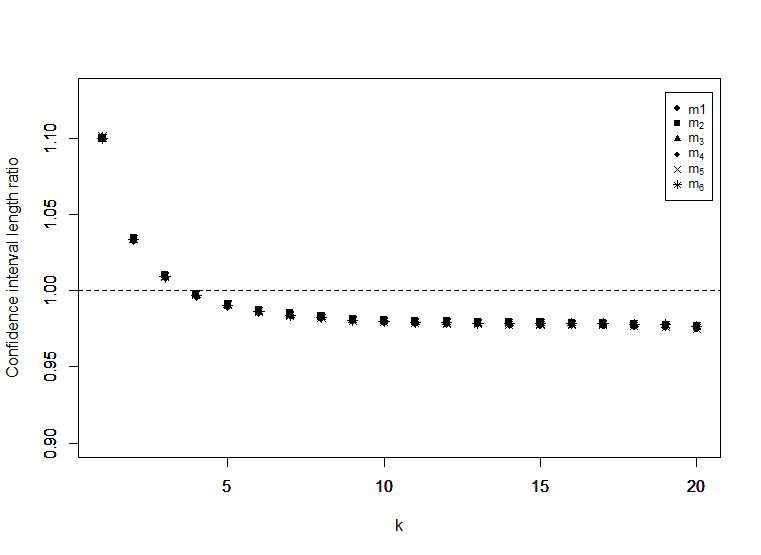}
\caption{95\% confidence interval length ratio binary outcome: Average treatment effects. Small $k$ implies relative wide intervals, whereas large $k$ implies narrow intervals. The dotted line is optimal selection of $k$ minimizing mean squared error. By construction has a ratio equal to 1.}
  \label{fig:IWRCTbinATE}
\end{figure}

\begin{figure}
\centering
  \includegraphics[scale=0.5]{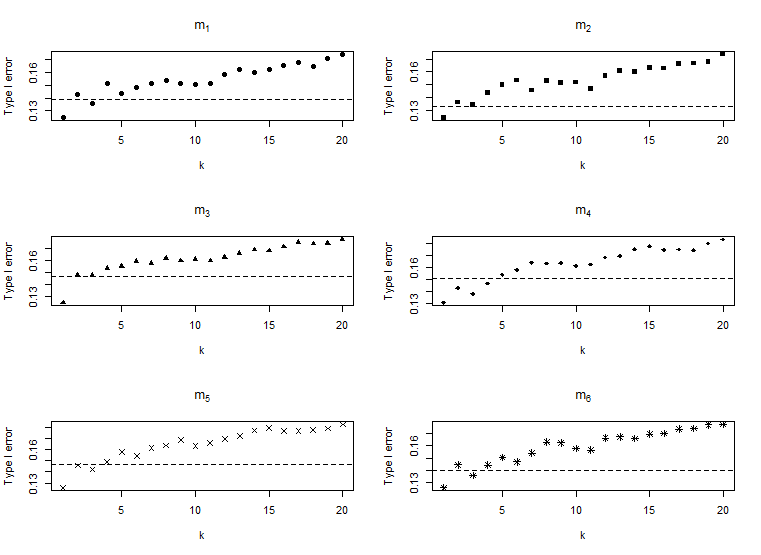}
\caption{Type I errors (5\% nominal size) binary outcome: Average treatment effects. The optimal $k$ gives a relatively small type I error compared with other $k$ selection. $k=1$ always gets the best nominal size.}
  \label{fig:TypeICTbinATE}
\end{figure}

\end{document}